\documentclass[aps,prd,twocolumn]{revtex4-1}
\usepackage{euscript}
\usepackage{amssymb}
\usepackage{graphicx}
\usepackage{amsmath}
\usepackage{amsbsy}
\usepackage{amsthm}
\newtheorem{theorem}{Theorem}
\usepackage{bbm}
\usepackage{bm}
\usepackage{epsfig}
\usepackage{epstopdf}
\usepackage{dsfont}
\usepackage{color}
\usepackage[colorlinks]{hyperref}
\usepackage[figure,table]{hypcap}
\usepackage[symbol]{footmisc}
\usepackage{enumerate}
\usepackage{geometry}
\hypersetup{
	bookmarksnumbered,
	pdfstartview={FitH},
	citecolor={darkgreen},
	linkcolor={darkblue},
	urlcolor={darkblue},
	pdfpagemode={UseOutlines}}
\definecolor{darkgreen}{RGB}{50,190,50}
\definecolor{darkblue}{RGB}{0,0,190}
\definecolor{darkred}{RGB}{238,0,0}
\usepackage{soul}

\newcommand{\tr}{\textnormal{Tr}}
\newcommand{\bra}[1]{\ensuremath{\left\langle\right. #1 \left.\right|}}
\newcommand{\ket}[1]{\ensuremath{\left|\right. #1 \left.\right\rangle}}
\newcommand{\fbra}[1]{\ensuremath{\left\langle\!\hspace*{-0.7pt}\left\langle\right.\right.\! #1 \!\left.\left.\right|\!\right|}}
\newcommand{\fket}[1]{\ensuremath{\left|\!\left|\right.\right.\! #1 \!\left.\left.\right\rangle\!\hspace*{-0.7pt}\right\rangle}}

\newcommand{\anticomm}[2]{\ensuremath{\left\{\right.\! #1 \,, #2 \!\left.\right\}}}
\newcommand{\scpr}[2]{\ensuremath{\left\langle\right. #1 \,\left|\right. #2 \left.\right\rangle}}
\newcommand{\fscpr}[2]{\ensuremath{\left\langle\!\hspace*{-0.7pt}\left\langle\right.\right.\! #1 \!\left.\left.\right|\!\right| #2 \!\left.\left.\right\rangle\!\hspace*{-0.7pt}\right\rangle}}
\newcommand{\pr}{^{\prime}}
\newcommand{\prpr}{^{\prime\hspace*{-0.5pt}\prime}}

\DeclareMathOperator{\spectr}{spectr}

\begin{document}

\title{Fermionic mode entanglement in quantum information
}

\date{December 2012}

\author{Nicolai Friis$^{1}$}
\email{pmxnf@nottingham.ac.uk}
\author{Antony R.~Lee$^{1}$}
\email{pmxal3@nottingham.ac.uk}
\author{David Edward Bruschi$^{1,2}$}
\email{david.edward.bruschi@gmail.com}
\affiliation{
$^{1}$School of Mathematical Sciences,
University of Nottingham,
University Park,
Nottingham NG7 2RD,
United Kingdom}
\affiliation{
$^{2}$School of Electronic and Electrical Engineering,
University of Leeds,
Woodhouse Lane,
Leeds LS2 9JT,
United Kingdom}

\begin{abstract}
We analyze 
fermionic modes as fundamental entities for quantum
information processing. To this end we construct a density operator formalism on
the underlying Fock space and demonstrate how it can be naturally and unambiguously
equipped with a notion of subsystems in the absence of a global tensor product
structure. We argue that any apparent similarities between fermionic modes and qubits
are superficial and can only be applied in limited situations. In particular, we
discuss the ambiguities that arise from different treatments of this subject. Our
results are independent of the specific context of the fermionic fields as long as
the canonical anticommutation relations are satisfied, e.g., in relativistic quantum
fields or fermionic trapped ions.
\end{abstract}

\pacs{
03.65.Ud,   
11.10.-z   
}
\maketitle

\section{Introduction}\label{sec:intro}

Fermionic systems have been analyzed as agents for quantum information processing in a multitude of
studies, ranging from discussions of fermionic modes of relativistic quantum fields~\cite{Shi2004,AlsingFuentes-SchullerMannTessier2006,FuentesMannMartin-MartinezMoradi2010,
BruschiLoukoMartin-MartinezDraganFuentes2010,Martin-MartinezFuentes2011,FriisKoehlerMartinMartinezBertlmann2011,
SmithMann2012,FriisLeeBruschiLouko2012,FriisBruschiLoukoFuentes2012,FriisHuberFuentesBruschi2012,
MonteroMartin-Martinez2012b}, over
fermionic lattices~\cite{Zanardi2002}, and fermionic Gaussian states~\cite{BoteroReznik2004}, to discussions of the entanglement between fixed
numbers of indistinguishable particles~\cite{SchliemannLossMacDonald2000,SchliemannCiracKusLewensteinLoss2001,PaskauskasYou2001,
LiZengLiuLong2001,EckertSchliemannBrussLewenstein2002,Shi2003,WisemanVaccaro2003,WisemanBartlettVaccaro2003,GhirardiMarinatto2004,
CabanPodlaskiRembielinskiSmolinskiWalczak2005,IeminiVianna2012}. In the latter case, only pure states of fixed particle numbers
are considered and a selection of entanglement measures are available, see, e.g., Ref.~\cite{WisemanVaccaro2003}. However, these
restrictions seem to be much more limiting than required. From the point of view of quantum information theory it is natural to ask for
an extension to incoherent mixtures of quantum states. Furthermore, from the perspective of a relativistic description particle numbers
are not usually conserved, i.e., the particle content of a given pure state is observer dependent~\cite{BirrellDavies:QFbook}. The
description of fermionic entanglement should therefore include coherent and incoherent mixtures of different particle numbers. Any
required superselection rules, e.g., for (electric) charge~\cite{StrocchiWightman1974} or parity, can then be considered as special cases of
such a framework.

In the light of this fact it is therefore reasonable to consider the entanglement between fermionic modes, in a similar way as is
conventionally done for bosonic modes, e.g., for Gaussian states~\cite{AdessoIlluminati2005}. We shall show here that the entanglement
of a system of fermionic modes can be defined unambiguously by enforcing a physically reasonable definition of its subsystems. This procedure
is completely independent of any superselection rules.

A central question that appears in practical situations is: \emph{Can fermionic modes be considered as qubits?} The short answer to this question is ``No." Due to the Pauli exclusion principle, fermionic modes are naturally restricted to two degrees of freedom, i.e., each mode can be unoccupied or contain a single excitation. This has provided many researchers with an ad hoc justification for the comparison with qubits -- two-level systems used in quantum information, which has incited debates among scientists, see, e.g., the exchange in Refs.~\cite{MonteroMartin-Martinez2011b,BradlerJauregui2012,MonteroMartin-Martinez2012a,Bradler2011}. In limited situations certain techniques from the study of qubits can indeed be applied to fermionic systems. However, while mappings between fermionic systems and qubits are possible in principle, e.g., via the Jordan-Wigner transformation~\cite{BanulsCiracWolf2007}, the problem lies in the consistent mapping between the subsystems. In the following we shall give a more precise answer to the question above, along with a detailed description of the problem.

Any superselection rules further restrict the possible operations that can be performed on single-mode subsystems, and it was argued that this should lead to a modified definition of the entanglement between modes~\cite{WisemanVaccaro2003}. At least for fixed particle content this problem can be circumvented~\cite{HeaneyVedral2009}. Moreover, even if quantum correlations are not directly accessible, a transfer of the entanglement to systems that are not encumbered by such restrictions should be possible, thus justifying the use of unmodified measures for mode entanglement.

The main aim of this paper is establishing a clear framework for the implementation of fermionic field modes as vessels for quantum information
tasks. To this end we present an analysis of the problem at hand, i.e., how the modes in a fermionic Fock space can be utilized as subsystems for quantum information processing. We present a framework that is based on simple physical requirements in which this can be achieved.
We further discuss the issues and restrictions in mapping fermionic modes to qubits and we show how previous work and proposed solutions, e.g., invoking superselection rules~\cite{BradlerJauregui2012}, fit into this framework.

The article is structured as follows: We start with a brief discussion of the description of fermionic Fock spaces in Sec.~\ref{sec:fermionic Fock space} and how density operators are constructed on such spaces in Sec.~\ref{sec:fermionic density operators}. We then go on to formulate the ``fermionic ambiguity'' that has been pointed out in Ref.~\cite{MonteroMartin-Martinez2011b} in Sec.~\ref{sec:the fermionic ambiguity}. Subsequently, we reinterpret this as an ambiguity in the definition of mode subsystems, which can be resolved by
physical consistency conditions, in Sec.~\ref{sec:partial trace ambiguity}. Finally, we discuss the implications for the quantification of entanglement between two fermionic modes in Sec.~\ref{sec:entanglement of fermionic modes}, before we investigate situations beyond two modes in Sec.~\ref{sec:fermionic entanglement beyond 2 modes}.

\section{The fermionic Fock space}\label{sec:fermionic Fock space}

Let us consider a (discrete) set of solutions $\psi_{n}$ to a (relativistic) field equation, e.g., the
Dirac equation. In the decomposition of the quantum field $\psi$, each mode function $\psi_{n}$ is
assigned an annihilation operator $b_{n}$ and a creation operator $b_{n}^{\dagger}$, which satisfy the
canonical anticommutation relations
\begin{subequations}
\label{eq:anticomm relations}
    \begin{align}
        \anticomm{b_{m}}{b_{n}^{\dagger}} &=\,\delta_{mn}\,,
        \label{eq:anticomm non-vanishing}\\[1.5mm]
        \anticomm{b_{m}}{b_{n}} &=\,\anticomm{b_{m}^{\dagger}}{b_{n}^{\dagger}}\,=\,0\,,
        \label{eq:anticomm vanishing}
    \end{align}
\end{subequations}
where $\anticomm{.}{.}$ denotes the anticommutator. One may introduce a different notation for the annihilation
and creation operators for modes with opposite charge (see, e.g., Ref.~\cite{FriisLeeBruschiLouko2012}), but for the
purpose of our analysis here this is inconsequential and we can work only with $b_{n}$ and $b_{n}^{\dagger}$.
The creation operators $b_{n}^{\dagger}$, acting upon the vacuum state $\ket{0}$, will populate the vacuum with a single excitation, i.e.,
\begin{align}
    \ket{\!\psi_{n}\!}   &=  \,b_{n}^{\dagger}\,\ket{0}\,,
    \label{eq:single fermion excitation}
\end{align}
while the vacuum is annihilated by all $b_{n}$, i.e., $b_{n}\ket{0}=0\ \forall\,n$. As can be quickly seen
from this property and Eq.~(\ref{eq:anticomm non-vanishing}), the states $\ket{\!\psi_{n}\!}$ are orthonormal.
The states $\ket{\!\psi_{n}\!}$ further form a complete basis of the single-particle Hilbert space
$\mathcal{H}_{1-p}$, whereas $\ket{0}\in\mathcal{H}_{0-p}\neq\mathcal{H}_{1-p}$.
A general state in $\mathcal{H}_{1-p}$ has the form
\begin{align}
    \ket{\!\psi^{1-p}\!}    &=\,\sum\limits_{i}\,\mu_{i}\,\ket{\!\psi_{i}\!}\,,
    \label{eq:general single fermion state}
\end{align}
with $\sum_{i}|\mu_{i}|^{2}=1$ such that $\scpr{\psi^{1-p}}{\psi^{1-p}}=1$. Let us now turn to states of multiple
fermions. A second fermion can be added to the state~(\ref{eq:single fermion excitation}) by the action of another
creation operator $b_{m}^{\dagger}$, i.e.,
\begin{align}
    b_{m}^{\dagger}b_{n}^{\dagger}\,\ket{0} &\propto\,\ket{\!\psi_{m},\psi_{n}\!}\,.
    \label{eq:two fermion excitation}
\end{align}
Clearly, the anticommutation relations (\ref{eq:anticomm relations}) require the two-fermion state to be
antisymmetric with respect to the exchange of the mode labels $m$ and $n$. We therefore define
\begin{align}
    \ket{\!\psi_{m},\psi_{n}\!} &=\,b_{m}^{\dagger}b_{n}^{\dagger}\,\ket{0}\,=\,
    \ket{\!\psi_{m}\!}\wedge\ket{\!\psi_{n}\!}
    \label{eq:two fermion state definition}
    \\[1.0mm]
    &=\,
    \tfrac{1}{\sqrt{2}}\Bigl(\ket{\!\psi_{m}\!}\otimes\ket{\!\psi_{n}\!}\,-\,\ket{\!\psi_{n}\!}\otimes\ket{\!\psi_{m}\!}\Bigr)\,.
    \nonumber
\end{align}
The two-fermion states are thus elements of the antisymmetrized tensor product space of two single-fermion Hilbert spaces, i.e.,
\begin{align}
    \mathcal{H}_{2-p}   &=\,\bar{S}\Bigl(\mathcal{H}_{1-p}\otimes\mathcal{H}_{1-p}\Bigr)\,,
    \label{eq:two fermion space}
\end{align}
and a general state within this space can be written as
\begin{align}
    \ket{\!\psi^{2-p}\!}    &=\,\sum\limits_{i,j}\,\mu_{ij}\,\ket{\!\psi_{i},\psi_{j}\!}\,,
    \label{eq:general two fermion state}
\end{align}
where the coefficients $\mu_{ij}$ form an antisymmetric matrix. States with more than two fermions can then be
constructed by antisymmetrizing over the corresponding number of single-fermion states. Finally, the $n$-mode \emph{fermionic
Fock space} $\bar{\mathcal{F}}_{n}$ is simply given as the direct sum over all fermion numbers of the antisymmetrized
Hilbert spaces, i.e.,
\begin{align}
    \bar{\mathcal{F}}_{n}(\mathcal{H}_{1-p})    &=\,
    \bigoplus\limits_{m=1}^{n}\bar{S}\Bigl(\mathcal{H}_{1-p}^{\otimes m}\Bigr)
    \label{eq:fermionic Fock space}\\[1.0mm]
    &=\,
    \mathcal{H}_{0-p}\oplus\mathcal{H}_{1-p}\oplus\bar{S}\Bigl(\mathcal{H}_{1-p}\otimes\mathcal{H}_{1-p}\Bigr)\oplus\ldots\,,
    \nonumber
\end{align}
where $
    \mathcal{H}^{\otimes m} 
$ denotes the $m$-fold tensor product and we write $\mathcal{H}_{0-p}$ as $\mathcal{H}_{1-p}^{\otimes 0}$. A general state in the space $\bar{\mathcal{F}}_{n}$ can be written as
\begin{align}
    \ket{\Psi^{\bar{\mathcal{F}}_{n}}}    &=\,
    \mu_{0}\ket{0}\,\oplus\,\sum\limits_{i=1}^{n}\mu_{i}\ket{\!\psi_{i}\!}
    \label{eq:fermionic Fock space general state}\\[1.0mm]
    &\ \oplus
    \sum\limits_{j,k}\mu_{jk}\ket{\!\psi_{j}\!}\wedge\ket{\!\psi_{k}\!}\,\oplus\,\ldots\,.
    \nonumber
\end{align}
Let us now simplify the notation. From now on we will denote states in the fermionic
Fock space by double-lined Dirac notation, i.e., $\fket{.}$ instead of $\ket{.}$, where the antisymmetric ``wedge'' product
is implied when two vectors are multiplied, i.e., $\fket{.}\fket{.}=\fket{.}\wedge\fket{.}$. Furthermore, let us use the
common ``occupation number'' notation and write $1_{n}$ instead of $\psi_{n}$ to denote an excitation in the mode $n$. Finally,
we omit the symbol for the direct sum and simply keep in mind that states with different numbers of excitations occupy different
sectors of the fermionic Fock space. With this convention in mind we can rewrite Eq.~(\ref{eq:fermionic Fock space general state}) as
\begin{align}
    \fket{\Psi}    &=\,
    \mu_{0}\fket{0}\,+\,\sum\limits_{i=1}^{n}\mu_{i}\fket{\!1_{i}\!}
    \label{eq:fermionic Fock space general state rewritten}\\[1.0mm]
    &+\
    \sum\limits_{j,k}\mu_{jk}\fket{\!1_{j}\!}\fket{\!1_{k}\!}\,+\,\ldots\ .
    \nonumber
\end{align}
For the adjoint space we use the convention [compare to Eq.~(\ref{eq:two fermion state definition})]
\begin{align}
    \fbra{\!1_{n}\!}\fbra{\!1_{m}\!}    &:=\,
    \fbra{0}b_{n}b_{m}\,=\,\Bigl(b_{m}^{\dagger}b_{n}^{\dagger}\,\fket{0}\Bigr)^{\dagger}
    \label{eq:adjoint space convention}\\[1.0mm]
    &\ =\,-\,
    \tfrac{1}{\sqrt{2}}\Bigl(\bra{\!\psi_{n}\!}\otimes\bra{\!\psi_{m}\!}\,-\,\bra{\!\psi_{m}\!}\otimes\bra{\!\psi_{n}\!}\Bigr)\,,
    \nonumber
\end{align}
which allows us to write
\begin{align}
    \fbra{\!1_{m}\!}\fscpr{1_{n}}{1_{i}}\fket{\!1_{j}\!}    &=\,
    \delta_{ni}\delta_{mj}\,-\,\delta_{nj}\delta_{mi}\,.
    \label{eq:fermion state normalization}
\end{align}
This notation is more convenient for computations in the fermionic Fock space. It should be noted that, in standard quantum information notation,
the position of a ``ket'' corresponds to a particular ordering of the subspaces with respect to the tensor product structure of the total space.
Here, however, there is no tensor product structure corresponding to different modes according to which the vectors $\fket{.}$ can be naturally ordered.

\section{Density operators in the fermionic Fock space}\label{sec:fermionic density operators}

In complete analogy to the usual case of mixed states on tensor product spaces we can now construct incoherent mixtures of pure state in a
fermionic Fock space. Let us first consider the projector on the state $\fket{\Psi}$ from Eq.~(\ref{eq:fermionic Fock space general state rewritten}), i.e.,
\begin{align}
    \fket{\Psi}\!\fbra{\Psi}  &=\,
    |\mu_{0}|^{2}\fket{0}\!\fbra{0}\,+\,
    \sum\limits_{i,i\pr}\mu_{i}\,\mu_{i\pr}^{*}\,\fket{\!1_{i}\!}\!\fbra{\!1_{j}\!}
    \nonumber\\[1.0mm]
    &\ +\,\sum\limits_{j,j\pr,k,k\pr}\mu_{jk}\,\mu_{j\pr k\pr}^{*}\,\fket{\!1_{j}\!}\fket{\!1_{k}\!}\!\fbra{\!1_{j\pr}\!}\fbra{\!1_{k\pr}\!}
    \nonumber\\[1.0mm]
    &\ +\,
    \sum\limits_{i}\Bigl(\mu_{i}\,\mu_{0}^{*}\,\fket{\!1_{i}\!}\!\fbra{0}\,+\,\mathrm{H.c.}\Bigr)\,+\,\ldots\, .
    \label{eq:projector on general fermion state}
\end{align}
Let us check that such an object satisfies the criteria for a density operator:
\begin{enumerate}[(i)]
\item{It can be immediately noticed that~(\ref{eq:projector on general fermion state}) provides a \emph{Hermitean} operator.}
\item{The \emph{normalization}, i.e., $\tr\bigl(\fket{\Psi}\!\fbra{\Psi}\bigr)=1$, is guaranteed by the normalization of $\fket{\Psi}$. In other words,
the trace of~(\ref{eq:projector on general fermion state}) is well defined and independent of the chosen (complete, orthonormal) basis in $\bar{\mathcal{F}}$.}
\item{\emph{Positivity}: Finally, the eigenvalues of $\fket{\Psi}\!\fbra{\Psi}$ are well defined, i.e.,~(\ref{eq:projector on general fermion state}) can be represented
as a diagonal matrix with diagonal entries $\{1,0,0,\ldots\}$, which clearly is a positive semidefinite spectrum.}
\end{enumerate}
We can then simply form incoherent mixtures of such pure states using convex sums, i.e.,
\begin{align}
    \varrho &=\,\sum\limits_{n}p_{n}\fket{\Psi_{n}}\!\fbra{\Psi_{n}}\,
    \label{eq:fermionic mixed states}
\end{align}
where $\sum_{n}p_{n}=1$, to construct the elements of the \emph{Hilbert--Schmidt space} $\mathcal{H}_{S}(\bar{\mathcal{F}})$ over the fermionic Fock space. Properties $(i)$ and $(ii)$ can trivially be seen to be satisfied for such \emph{mixed} states. The positivity of~(\ref{eq:fermionic mixed states}), however,
requires some additional comments. The operator $\varrho$ can be diagonalized by a unitary transformation $U$ on $\bar{\mathcal{F}}$, which in turn can be
constructed from exponentiation of Hermitean or anti-Hermitean operators formed from algebra elements $b_{n}$ and $b_{m}^{\dagger}$. Operationally this procedure is rather elaborate. A simpler approach is the diagonalization of a matrix representation of $\varrho$. As we shall see in Sec.~\ref{sec:the fermionic ambiguity}, the matrix representation of $\varrho$ is not unique, but all possible representations $\pi_{i}(\varrho)$ are unitarily equivalent, such that their
eigenvalues all coincide with those of $\varrho$, i.e.,
\begin{align}
    \spectr\bigl(\pi_{i}(\varrho)\bigr) &=\,\spectr\bigl(\varrho\bigr)\ \ \forall\,i\,.
    \label{eq:spectrum of fermionic density matrix representation}
\end{align}

\section{The fermionic ambiguity}\label{sec:the fermionic ambiguity}

Let us now turn to the apparent ambiguity in such fermionic systems when quantum information tasks are considered. It was pointed out
in Ref.~\cite{MonteroMartin-Martinez2011b} that the anticommutation relations~(\ref{eq:anticomm relations}) do not suggest a natural
choice for the basis vectors of the fermionic Fock space for the multiparticle sector, i.e., for two fermions in the modes $m$ and $n$,
either
\begin{align}
    \fket{\!1_{m}\!}\fket{\!1_{n}\!}\ \ \ \mbox{or}\ \ \ \fket{\!1_{n}\!}\fket{\!1_{m}\!}\,=\,-\,\fket{\!1_{m}\!}\fket{\!1_{n}\!}
    \label{eq:fermionic ambiguity basis}
\end{align}
could be used to represent the physical state. This becomes of importance when we try to map the states in a fermionic $n$-mode Fock space to vectors in an $n$-fold tensor product space, i.e.,
\begin{subequations}
\label{eq:fermionic qubit mapping}
\begin{align}
    \pi_{i}:\ \bar{\mathcal{F}}_{n}\ &\longrightarrow\ \mathcal{H}_{1}\otimes\ldots\otimes\mathcal{H}_{n}\,
    \label{eq:fermionic qubit mapping spaces}\\[1.5mm]
    \fket{\psi}\ &\stackrel{\pi_{i}}{\longmapsto}\ \ket{\psi_{(i)}}
    \label{eq:fermionic qubit mapping pure states}\\[1.5mm]
    \varrho\ &\stackrel{\pi_{i}}{\longmapsto}\ \pi_{i}(\varrho)
    \label{eq:fermionic qubit mapping mixed states}
\end{align}
\end{subequations}
where the spaces $\mathcal{H}_{i}=\mathbb{C}^{2}$~$(i=1,\ldots,n)$ are identical, single-qubit Hilbert spaces. The mappings $\pi_{i}$ are unitary, i.e., $\fscpr{\phi}{\psi}=\scpr{\phi_{(i)}}{\psi_{(i)}}$ and $\tr(\varrho\sigma)=\tr(\pi_{i}(\varrho)\pi_{i}(\sigma))$. This implies that the maps $\pi_{i}$ for different $i$ are unitarily equivalent. In particular, the different matrix representations $\pi_{i}(\varrho)$ are related by multiplication of selected rows and columns of the matrix by $(-1)$.

In the language of quantum information theory the states $\psi_{(i)}$ are related by \emph{global unitary} transformations. It thus becomes apparent that the entanglement of $\pi_{i}(\varrho)$ with respect to a bipartition
\begin{align}
\mathcal{H}_{\mu_{1}}\otimes\ldots\otimes\mathcal{H}_{\mu_{m}}|\mathcal{H}_{\mu_{m+1}}\otimes\ldots\otimes\mathcal{H}_{\mu_{n}}
    \label{eq:bipartition}
\end{align}
will generally depend on the chosen mapping.

Clearly, this is an unfavorable situation, but the inequivalence of entanglement measures for different such mappings has been noted before (see, e.g., Refs.~\cite{BoteroReznik2004,CabanPodlaskiRembielinskiSmolinskiWalczak2005,BradlerJauregui2012}), while other investigations~\cite{FriisLeeBruschiLouko2012,FriisBruschiLoukoFuentes2012,FriisHuberFuentesBruschi2012} did not suffer from any problems due to this ambiguity. Recently, the authors of Ref.~\cite{BradlerJauregui2012} suggested that the ambiguity can be resolved by restrictions imposed by charge superselection rules, while Refs.~\cite{MonteroMartin-Martinez2011b,MonteroMartin-Martinez2012b} suggested a solution by enforcing a particular operator ordering. We will discuss both of these approaches in Sec.~\ref{sec:partial trace ambiguity}, where we present simple and physically intuitive criteria for quantum information processing on a fermionic Fock space.
Most importantly, we will show in Secs.~\ref{sec:partial trace ambiguity} and~\ref{sec:fermionic entanglement beyond 2 modes} that mappings of the type of (\ref{eq:fermionic qubit mapping}) can only be considered to be consistent when limiting the analysis to two fermionic modes obeying charge superselection, but not beyond this regime.

\section{The partial trace ambiguity}\label{sec:partial trace ambiguity}

While the sign ambiguity in the sense of the different mappings $\pi_{i}$ is the superficial cause of the issue, we want to
discuss now a separate, and in some sense more fundamental problem: partial traces over ``mode subspaces." We are interested
in the entanglement between modes of a fermionic quantum field. However, in the structure of the Fock space, there is no
tensor product decomposition into Hilbert spaces for particular modes [see, e.g., Eq.~(\ref{eq:two fermion state definition})].
Only a tensor product structure with respect to individual fermions is available, but since the particles are indistinguishable,
the entanglement between two particles in this sense has to be defined very carefully~\cite{WisemanVaccaro2003}. This issue is not
unique for fermions and is sometimes referred to as ``fluffy bunny'' entanglement (see Ref.~\cite{WisemanBartlettVaccaro2003}).

For the decomposition into different modes we only have a wedge product structure available. In Ref.~\cite{BradlerJauregui2012}
the authors suggest that entanglement should be considered with respect to this special case of the ``braided tensor product."
As far as the construction of the density operators with respect to such a structure is concerned, we agree with this view (see
Sec.~\ref{sec:fermionic density operators}), and no ambiguities arise regarding the description of the total $n$-mode system.
However, the crucial problem lies in the definition of the partial tracing over a subset of the $n$~modes. This is best
illustrated for a simple example: Consider a system of two fermionic modes labelled $\kappa$ and $\kappa\pr$. A general, mixed
state of these two modes can be written as
\begin{align}
    \varrho_{\kappa\kappa\pr}   &=\,
    \alpha_{1}\,\fket{0}\!\fbra{0}\,+\,
    \alpha_{2}\,\fket{\!1_{\kappa\pr}\!}\!\fbra{\!1_{\kappa\pr}\!}
    \label{eq:general two-mode fermion state}\\[1.5mm]
    &\ +\,
    \alpha_{3}\,\fket{\!1_{\kappa}\!}\!\fbra{\!1_{\kappa}\!}\,+\,
    \alpha_{4}\,\fket{\!1_{\kappa}\!}\fket{\!1_{\kappa\pr}\!}\!\fbra{\!1_{\kappa\pr}\!}\fbra{\!1_{\kappa}\!}
    \nonumber\\[1.0mm]
    &\ +\,
    \Bigl(
    \beta_{1}\,\fket{0}\!\fbra{\!1_{\kappa\pr}\!}\,+\,
    \beta_{2}\,\fket{0}\!\fbra{\!1_{\kappa}\!}
    \nonumber\\[1.0mm]
    &\ +\,
    \beta_{3}\,\fket{0}\!\fbra{\!1_{\kappa\pr}\!}\fbra{\!1_{\kappa}\!}\,+\,
    \beta_{4}\,\fket{\!1_{\kappa\pr}\!}\!\fbra{\!1_{\kappa}\!}
    \nonumber\\[1.5mm]
    &\ +\,
    \beta_{5}\,\fket{\!1_{\kappa\pr}\!}\!\fbra{\!1_{\kappa\pr}\!}\fbra{\!1_{\kappa}\!}\,+\,
    \beta_{6}\,\fket{\!1_{\kappa}\!}\!\fbra{\!1_{\kappa\pr}\!}\fbra{\!1_{\kappa}\!}
    \nonumber\\[1.0mm]
    &\ +\,\mathrm{H.c.}\,\Bigr)\,,
    \nonumber
\end{align}
where appropriate restrictions on the coefficients $\alpha_{i}\in\mathbb{R}$ and $\beta_{j}\in\mathbb{C}$ apply to ensure the positivity
and normalization of $\varrho_{\kappa\kappa\pr}$. Here we have, for now, disregarded superselection rules. Let us now determine
the corresponding reduced density operators (on the Fock space) for the individual modes $\kappa$ and $\kappa\pr$. Usually one would
select a basis of the subsystem that is being traced over, e.g., for tracing over mode $\kappa\pr$ one could choose
$\{\fket{0},\fket{\!1_{\kappa\pr}\!}\}$. This clearly cannot work since basis vectors with different numbers of excitations are orthogonal. We thus have to define the partial trace in a different way. This is equally true for bosonic fields as well. However, in contrast to the fermionic case, no ambiguities arise in such a redefinition for bosonic fields. For the diagonal elements of the reduced fermionic states the redefinition of the partial trace is straightforward as well. These elements are obtained from
\begin{subequations}
\label{eq:partial trace diagonal elements}
\begin{align}
    \tr_{m}\bigl(\fket{0}\!\fbra{0}\bigr) &:=\,\fket{0}\!\fbra{0}\,,
    \label{eq:partial trace diagonal element vacuum}\\[1.0mm]
    \tr_{m}\bigl(\fket{\!1_{n}\!}\!\fbra{\!1_{n}\!}\bigr) &:=\,\delta_{mn}\fket{\!1_{n}\!}\!\fbra{\!1_{n}\!}
    \label{eq:partial trace diagonal element single particle}\\[1.0mm]
    &\ +\,(1-\delta_{mn})\fket{0}\!\fbra{0}\,,
    \nonumber\\[1.0mm]
    \tr_{m}\bigl(\fket{\!1_{m}\!}\fket{\!1_{n}\!}\!\fbra{\!1_{n}\!}\fbra{\!1_{m}\!}\bigr) &:=\,\fket{\!1_{n}\!}\!\fbra{\!1_{n}\!}\
    (m\neq n)\,,
    \label{eq:partial trace diagonal element two particles}
\end{align}
\end{subequations}
where $n,m=\kappa,\kappa\pr$. While the diagonal elements are unproblematic and do not suffer from any ambiguities, we have to be more careful with the off-diagonal elements. Three of these will not contribute, i.e.,
\begin{align}
    \tr_{m}\bigl(\fket{\!1_{m}\!}\!\fbra{\!1_{n}\!}\bigr) &=\,
    \tr_{m}\bigl(\fket{0}\!\fbra{\!1_{m}\!}\fbra{\!1_{n}\!}\bigr)
    \label{eq:partial trace offdiagonal elements vanishing}\\[1.5mm]
    &=\,
    \tr_{m}\bigl(\fket{\!1_{n}\!}\!\fbra{\!1_{m}\!}\fbra{\!1_{n}\!}\bigr)\,=\,0\,,
    \nonumber
\end{align}
and two more are unproblematic as well, i.e.,
\begin{align}
    \tr_{m}\bigl(\fket{0}\!\fbra{\!1_{n}\!}\bigr) &:=\,(1-\delta_{mn})\fket{0}\!\fbra{\!1_{n}\!}\,.
    \label{eq:partial trace offdiagonal elements trivial}
\end{align}
The last element,
\begin{align}
    \tr_{m}\bigl(\fket{\!1_{m}\!}\!\fbra{\!1_{m}\!}\fbra{\!1_{n}\!}\bigr)   &=\,
    -\,\tr_{m}\bigl(\fket{\!1_{m}\!}\!\fbra{\!1_{n}\!}\fbra{\!1_{m}\!}\bigr)
    \nonumber\\[1.5mm]
    &=\,\pm\,\fket{0}\!\fbra{\!1_{n}\!}\,,
    \label{eq:partial trace offdiagonal ambiguous}
\end{align}
however, presents an ambiguity. If a mapping $\pi_{i}$ to a two-qubit Hilbert space is performed, the choice of map will
determine the corresponding sign in the partial trace over either of the qubits. The differences in entanglement related
to the fact that $\pi_{i}(\varrho)$ and $\pi_{j}(\varrho)$ are related by a global unitary are thus explained by the
relative sign between the contributions of Eq.~(\ref{eq:partial trace offdiagonal elements trivial}) and
Eq.~(\ref{eq:partial trace offdiagonal ambiguous}) to the same element of the reduced density matrix.

However, simple \emph{physical requirements} restrict the choice in this relative sign. Any reduced state formalism has
to satisfy the simple criterion that the reduced density operator contains all the information about the subsystem that
can be obtained from the global state when measurements are performed only on the respective subsystem alone.

Let us put this statement in more mathematical terms. For any bipartition $A|B$ of a Hilbert space~$\mathcal{H}$ (with respect to any
braided tensor product structure on~$\mathcal{H}$) and any state $\rho\in\mathcal{H}$ the partial trace operation $\tr_{B}$ must satisfy
\begin{align}
    \left\langle\,\mathcal{O}_{n}(A)\,\right\rangle_{\rho}  &=\,\left\langle\,\mathcal{O}_{n}(A)\,\right\rangle_{\tr_{B}(\rho)}\,,
    \label{eq:consistency condition}
\end{align}
where $\left\langle\mathcal{O}\right\rangle_{\rho}$ denotes the expectation value of the operator $\mathcal{O}$ in the state $\rho$
and $\{\mathcal{O}_{n}(A)\}$ is the set of all (Hermitean) operators that act on the subspace $A$ only. For the operator
$\varrho_{\kappa\kappa\pr}$ from Eq.~(\ref{eq:general two-mode fermion state}) the condition~(\ref{eq:consistency condition}) can be written as
\begin{align}
    \tr\bigl(\mathcal{O}_{n}(\kappa)\varrho_{\kappa\kappa\pr}\bigr)  &=\,
    \tr\bigl(\mathcal{O}_{n}(\kappa)\varrho_{\kappa}\bigr)\,,
    \label{eq:consistency condition 2 modes}
\end{align}
where $\varrho_{\kappa}=\tr_{\kappa\pr}(\varrho_{\kappa\kappa\pr})$. This consistency condition uniquely determines the relative signs
between different contributions to the same elements of $\varrho_{\kappa}$. Let us consider the (Hermitean) operators
$(b_{\kappa}+b_{\kappa}^{\dagger})$ and $i(b_{\kappa}-b_{\kappa}^{\dagger})$. Their expectation values for the global state
$\varrho_{\kappa\kappa\pr}$ are given by
\begin{subequations}
\label{eq:consistencty condition global exp values 2 modes kappa}
    \begin{align}
    \tr\bigl(\,(b_{\kappa}+b_{\kappa}^{\dagger})\varrho_{\kappa\kappa\pr}\,\bigr)   &=\,
    2\,\mathrm{Re}(\,\beta_{2}+\beta_{5}\,)\,,
    \label{eq:consistencty condition global exp values 2 modes kappa x}\\[1.0mm]
    \tr\bigl(\,i(b_{\kappa}-b_{\kappa}^{\dagger})\varrho_{\kappa\kappa\pr}\,\bigr)   &=\,
    2\,\mathrm{Im}(\,\beta_{2}+\beta_{5}\,)\,.
    \label{eq:consistencty condition global exp values 2 modes kappa p}
\end{align}
\end{subequations}
For the mode $\kappa\pr$, on the other hand, we compute
\begin{subequations}
\label{eq:consistencty condition global exp values 2 modes kappa prime}
    \begin{align}
    \tr\bigl(\,(b_{\kappa\pr}+b_{\kappa\pr}^{\dagger})\varrho_{\kappa\kappa\pr}\,\bigr)   &=\,
    2\,\mathrm{Re}(\,\beta_{1}-\beta_{6}\,)\,,
    \label{eq:consistencty condition global exp values 2 modes kappa pr x}\\[1.0mm]
    \tr\bigl(\,i(b_{\kappa\pr}-b_{\kappa\pr}^{\dagger})\varrho_{\kappa\kappa\pr}\,\bigr)   &=\,
    2\,\mathrm{Im}(\,\beta_{1}-\beta_{6}\,)\,.
    \label{eq:consistencty condition global exp values 2 modes kappa pr p}
    \end{align}
\end{subequations}
Equations~(\ref{eq:consistencty condition global exp values 2 modes kappa}) and
(\ref{eq:consistencty condition global exp values 2 modes kappa prime}) determine the sign in
Eq.~(\ref{eq:partial trace offdiagonal ambiguous}) and we find the reduced states
\begin{subequations}
\label{eq:reduced states for kappa and kappa prime}
    \begin{align}
        \varrho_{\kappa}\,=\,\tr_{\kappa\pr}\bigl(\varrho_{\kappa\kappa\pr}\bigr) &=\,
        (\alpha_{1}+\alpha_{2})\,\fket{0}\!\fbra{0}
        \label{eq:reduced state for kappa}\\[1.0mm]
        &+\,
        (\alpha_{3}+\alpha_{4})\,\fket{\!1_{\kappa}\!}\!\fbra{\!1_{\kappa}\!}
        \nonumber\\[1.0mm]
        &+\,
        \Bigl((\beta_{2}+\beta_{5})\,\fket{0}\!\fbra{\!1_{\kappa}\!}\,+\,\mathrm{H.c.}\Bigr)\,
        \nonumber\\[1.0mm]
        \varrho_{\kappa\pr}\,=\,\tr_{\kappa}\bigl(\varrho_{\kappa\kappa\pr}\bigr) &=\,
        (\alpha_{1}+\alpha_{3})\,\fket{0}\!\fbra{0}\label{eq:reduced state for kappa}\\[1.0mm]
        &+\,
        (\alpha_{2}+\alpha_{4})\,\fket{\!1_{\kappa\pr}\!}\!\fbra{\!1_{\kappa\pr}\!}
        \nonumber\\[1.0mm]
        &+\,\Bigl((\beta_{1}-\beta_{6})\,\fket{0}\!\fbra{\!1_{\kappa\pr}\!}\,+\,\mathrm{H.c.}\Bigr)
        \nonumber
    \end{align}
\end{subequations}
for the modes $\kappa$ and $\kappa\pr$, respectively. Notice that this formally corresponds to tracing \emph{``inside out,"}
that is, first (anti-)commuting operators towards the projector on the vacuum state before removing them, such that
\begin{align}
    \tr_{m}\bigl(b_{m}^{\dagger}\fket{0}\!\fbra{0}b_{m}b_{n}\bigr)   &=\,
    \fket{0}\!\fbra{\!1_{n}\!}\,.
    \label{eq:partial trace offdiagonal ambiguous inside out}
\end{align}
We have now arrived at a point where we can make a general statement about the consistency conditions. Let us formulate
this in the following theorem.

\begin{theorem}
Given a density operator $\varrho_{1,\ldots,n}\in\mathcal{H}_{S}(\bar{\mathcal{F}}_{n})$ for $n$~fermionic modes
(labelled $1,\ldots,n$) the consistency conditions~\rm{(\ref{eq:consistency condition})} completely determine the reduced
states on $\mathcal{H}_{S}(\bar{\mathcal{F}}_{m})$ for any $m$ with $1<m<n$.
\label{theorem:consistency condition}
\end{theorem}

\begin{proof}
This can be seen in the following way: for any matrix element
\begin{align}
\gamma\,b^{\dagger}_{\mu_{1}}\ldots b^{\dagger}_{\mu_{i}}\fket{0}\!\fbra{0}b_{\nu_{1}}\ldots b_{\nu_{j}}
\label{eq:general matrix element n-1 mode state}
\end{align}
of an $(n-1)$-mode reduced state $\varrho_{1,\ldots,(n-1)}=\tr_{n}(\varrho_{1,\ldots,n})$, where $\gamma\in\mathbb{C}$ and the sets
\begin{subequations}
\label{eq:n-1 mode subsets}
    \begin{align}
        \mu\,:=\,\{\mu_{1},\ldots,\mu_{i}\}  &\subseteq\,\{1,2,\ldots,(n-1)\}
        \label{eq:n-1 mode subset mu}\\[1.5mm]
    \mbox{and}\ \  \nu\,:=\,\{\nu_{1},\ldots,\nu_{j}\}   &\subseteq\,\{1,2,\ldots,(n-1)\}
        \label{eq:n-1 mode subset nu}
    \end{align}
\end{subequations}
label subsets of the mode operators for the $(n-1)$ modes, can have contributions from at most two matrix elements of $\varrho_{1,\ldots,n}$, i.e.,
\begin{subequations}
\label{eq:contributions to n-1 modes element from n mode state}
    \begin{align}
        &\ \tr_{n}\bigl(\gamma_{0}\,b^{\dagger}_{\mu_{1}}\ldots b^{\dagger}_{\mu_{i}}
            \fket{0}\!\fbra{0}b_{\nu_{1}}\ldots b_{\nu_{j}}\bigr)\,,
        \label{eq:contribution no particle to n-1 modes element from n mode state}\\[1.5mm]
    \mbox{and}\ \
        &\ \tr_{n}\bigl(\gamma_{1}\,b^{\dagger}_{\mu_{1}}\ldots b^{\dagger}_{\mu_{i}}b^{\dagger}_{n}
            \fket{0}\!\fbra{0}b_{n}b_{\nu_{1}}\ldots b_{\nu_{j}}\bigr)\,.
        \label{eq:contribution one particle to n-1 modes element from n mode state}
    \end{align}
\end{subequations}
The composition of $\gamma$ into $\gamma_{1}\in\mathbb{C}$ and $\gamma_{2}\in\mathbb{C}$, i.e., $\gamma=\gamma_{0}\pm\gamma_{1}$, is determined by the
consistency conditions of Eq.~(\ref{eq:consistency condition}). For every matrix element~(\ref{eq:general matrix element n-1 mode state}) with corresponding partial trace contributions from (\ref{eq:contributions to n-1 modes element from n mode state}) there exists a pair of Hermitean operators
\begin{subequations}
\label{eq:general element consistency condition operators}
\begin{align}
\mathcal{O}_{x}(\lambda,\tau)   &=\,b_{\lambda_{1}}\ldots b_{\lambda_{k}}b^{\dagger}_{\tau_{1}}\ldots b^{\dagger}_{\tau_{l}}
    \label{eq:general element consistency condition operator x}\\[1.0mm]
    &\ +\,
    b_{\tau_{l}}\ldots b_{\tau_{1}}b^{\dagger}_{\lambda_{k}}\ldots b^{\dagger}_{\lambda_{1}}\,,
    \nonumber\\[1.5mm]
\mathcal{O}_{p}(\lambda,\tau)   &=\,b_{\lambda_{1}}\ldots b_{\lambda_{k}}b^{\dagger}_{\tau_{1}}\ldots b^{\dagger}_{\tau_{l}}
    \label{eq:general element consistency condition operator p}\\[1.0mm]
    &\ -i\,
    b_{\tau_{l}}\ldots b_{\tau_{1}}b^{\dagger}_{\lambda_{k}}\ldots b^{\dagger}_{\lambda_{1}}\,,
    \nonumber
\end{align}
\end{subequations}
with $\lambda:=\{\lambda_{1},\ldots,\lambda_{k}\}=\mu/\nu$ and $\tau:=\{\tau_{1},\ldots,\tau_{l}\}=\nu/\mu$, that
uniquely determine the relative sign of $\gamma_{1}$ and $\gamma_{2}$.
These operators are unique up to an overall multiplication with scalars.
The tracing procedure can be repeated when any other of the $(n-1)$ remaining modes
are traced over. Since the order of the partial traces is of no importance for the final reduced state, all reduced
density operators are completely determined.
\end{proof}
Consequently, the reduced density matrices
in the fermionic Fock space can be considered as proper density operators, i.e., they are Hermitean, normalized, and
their eigenvalues are well defined and non-negative. Moreover, since the eigenvalues are free of ambiguities, all
functions of these eigenvalues, in particular, all entropy measures for density operators, are well defined. Also,
the operator ordering that was suggested in Ref.~\cite{MonteroMartin-Martinez2011b} is consistent with our consistency
condition.

Let us stress here that this analysis does not depend on any superselection rules that might be imposed in addition.
We will see how these enter the problem when mappings to qubits are attempted in Sec.~\ref{sec:entanglement of fermionic modes}.

\section{Entanglement of fermionic modes}\label{sec:entanglement of fermionic modes}

We are now in a position to reconsider a measure of entanglement between fermionic modes. We can define the entanglement of formation
$\bar{E}_{oF}$ for fermionic systems with respect to a chosen bipartition $A|B$ as
\begin{align}
    \bar{E}_{oF}(\varrho)    &=\,
    \min_{\{p_{n},\fket{\!\Psi_{n}\!}\}}\,\sum\limits_{n}\,p_{n}\,\mathcal{E}(\fket{\!\Psi_{n}\!})\,,
    \label{eq:fermionic entanglement of formation}
\end{align}
in complete analogy to the usual definition~\cite{BennettDiVincenzoSmolinWootters1996}. Here the minimum is taken over all pure
state ensembles $\fket{\!\Psi_{n}\!}$ that realize $\varrho$ according to
Eq.~(\ref{eq:fermionic mixed states}) and $\mathcal{E}(\fket{\!\Psi\!})$ denotes the entropy of entanglement of the
pure state $\fket{\!\Psi\!}$. Since the entropy of entanglement, e.g., using the von~Neumann entropy, is a function of the eigenvalues of the reduced states
$\tr_{B}\bigl(\fket{\!\Psi\!}\!\fbra{\!\Psi\!}\bigr)$ or $\tr_{A}\bigl(\fket{\!\Psi\!}\!\fbra{\!\Psi\!}\bigr)$ alone, we
can conclude that this is a well-defined quantity. As pointed out in Ref.~\cite{CabanPodlaskiRembielinskiSmolinskiWalczak2005},
the minimization in Eq.~(\ref{eq:fermionic entanglement of formation}) can be restricted to pure state decompositions that respect
superselection rules. Since this restriction limits the set of states over which the minimization is carried out, the quantity without this
restriction will be a lower bound to the ``physical'' entanglement of formation. For two fermionic modes the minimization over all states
that respect superselection rules can indeed be carried out (see Ref.~\cite{CabanPodlaskiRembielinskiSmolinskiWalczak2005}). However, in
general this step will be problematic.

Let us now turn to some operational entanglement measure, in particular, let us investigate if and how the \emph{negativity} $\mathcal{N}$ (see, e.g., Ref.~\cite{VidalWerner2001}) and the \emph{concurrence} $C$ (see, e.g., Ref.~\cite{BennettDiVincenzoSmolinWootters1996}) can be computed to
quantify fermionic mode entanglement. Both of these measures are operationally based on the tensor product structure of qubits. We will here
define the negativity as
\begin{align}
    \mathcal{N} &:=\,\sum\limits_{i}\frac{(\lambda_{i}-|\lambda_{i}|)}{2}\,,
    \label{eq:negativity}
\end{align}
where $\lambda_{i}$ are the eigenvalues of the partially transposed density matrix. However, the partial transposition is a map that is well defined only for basis vectors on a tensor product space. To employ this measure, let us therefore try to find a mapping $\pi_{i}$ to such a tensor product structure that is consistent with the conditions of Eq.~(\ref{eq:consistency condition}). Starting with the two-mode state $\varrho_{\kappa\kappa\pr}$ of Eq.~(\ref{eq:general two-mode fermion state}), we are
looking for a map $\pi$ that takes $\{\fket{0},\fket{\!1_{\kappa}\!},\fket{\!1_{\kappa\pr}\!},\fket{\!1_{\kappa}\!}\fket{\!1_{\kappa\pr}\!}\}$ to
$\{\ket{00},\ket{01\!},\ket{\!10},\ket{\!11\!}\}$, where $\ket{mn}=\ket{m}\otimes\ket{n}\in\mathcal{H}_{\kappa}\otimes\mathcal{H}_{\kappa\pr}$, such that
\begin{subequations}
\label{eq:two fermion mode consistent mapping}
    \begin{align}
    \varrho &\ \longmapsto\    \pi(\varrho)\,,
    \label{eq:two fermion mode consistent mapping total state}\\[1.0mm]
    \varrho_{\kappa} &\ \longmapsto\    \pi(\varrho_{\kappa})\,=\,\tr_{\kappa\pr}\bigl(\pi(\varrho)\bigr)\,,
    \label{eq:two fermion mode consistent mapping total state mode kappa}\\[1.0mm]
    \varrho_{\kappa\pr} &\ \longmapsto\    \pi(\varrho_{\kappa\pr})\,=\,\tr_{\kappa}\bigl(\pi(\varrho)\bigr)\,.
    \label{eq:two fermion mode consistent mapping total state mode kappa prime}
    \end{align}
\end{subequations}
The condition for a consistent mapping can be represented in the following diagram:
\begin{eqnarray}
\varrho_{\kappa\kappa\pr}\                        &\  \stackrel{\pi}{\longmapsto}\    &\ \pi(\varrho_{\kappa\kappa\pr})
\nonumber\\[2.0mm]
\tr_{\kappa\pr}\ \downarrow\ \ \ &                                   &\ \ \ \downarrow\ \ \tr_{\kappa\pr}
\label{eq:graphical representation of consitent map}\\[2.0mm]
\varrho_{\kappa}          \  \    &\  \stackrel{\pi}{\longmapsto}\    &\ \pi(\varrho_{\kappa})
\nonumber
\end{eqnarray}
In other words, a mapping $\pi: \varrho\mapsto\pi(\varrho)$ from the space $\mathcal{H}_{S}(\bar{\mathcal{F}}_{2})$ to $\mathcal{H}_{\kappa}\otimes\mathcal{H}_{\kappa\pr}$ is considered to be consistent if it commutes with the partial trace operation. It is quite simple to check that these requirements generally cannot be met, i.e., writing $\varrho_{\kappa\kappa\pr}$ of Eq.~(\ref{eq:general two-mode fermion state}) as a matrix with respect to the basis $\{\fket{0},\fket{\!1_{\kappa}\!},\fket{\!1_{\kappa\pr}\!},\fket{\!1_{\kappa}\!}\fket{\!1_{\kappa\pr}\!}\}$ we get
\begin{align}
\varrho &=\,
    \begin{pmatrix}
        \alpha_{1}      &   \beta_{1}       &   \beta_{2}       &   \beta_{3}   \\
        \beta_{1}^{*}   &   \alpha_{2}      &   \beta_{4}       &   \beta_{5}   \\
        \beta_{2}^{*}   &   \beta_{4}^{*}   &   \alpha_{3}      &   \beta_{6}   \\
        \beta_{3}^{*}   &   \beta_{5}^{*}   &   \beta_{6}^{*}   &   \alpha_{4}
    \end{pmatrix}\,.
\label{eq:general two mode state matrix representation}
\end{align}
A mapping of the desired type should be obtained by multiplying any number of rows and the corresponding columns by $(-1)$ and considering the resulting
matrix as the representation $\pi(\varrho)$ on $\mathcal{H}_{\kappa}\otimes\mathcal{H}_{\kappa\pr}$. The desired result should have a relative
sign switch between $\beta_{1}$ and $\beta_{6}$, while the signs in front of $\beta_{2}$ and $\beta_{5}$ should be the same. This clearly is not
possible unless some of the coefficients vanish identically, e.g., by imposing superselection rules. For example, conservation of charge would require the coefficients $\beta_{1},\beta_{2},\beta_{5}$, $\beta_{6}$, and, depending on the charge of the modes $\kappa$ and $\kappa\pr$, either $\beta_{3}$ or $\beta_{4}$ to vanish identically. In this way only incoherent mixtures of pure states with different charge are allowed, but not coherent superpositions.

We thus find that \emph{two fermionic modes can only be consistently represented as two qubits when charge superselection is respected}. In that case only one off-diagonal element can be nonzero and the sign of this element is insubstantial, i.e., it does not influence the reduced states or the value of any entanglement measure. In particular, the results for entanglement generation and degradation between two fermionic modes presented in Refs.~\cite{FriisLeeBruschiLouko2012,FriisBruschiLoukoFuentes2012,FriisKoehlerMartinMartinezBertlmann2011} respect both charge superselection and the consistency conditions of Eq.~(\ref{eq:consistency condition}).

Let us return to the choice of entanglement measure for the permitted mappings to two qubits. We now restrict the entanglement of formation $\bar{E}_{oF}$ as defined in Eq.~(\ref{eq:fermionic entanglement of formation}) to states that obey charge superselection, as suggested in Ref.~\cite{CabanPodlaskiRembielinskiSmolinskiWalczak2005}. As discussed earlier, this means the usual entanglement of formation $E_{oF}$ provides a
lower bound to $\bar{E}_{oF}$, i.e.,
\begin{align}
E_{oF}  &\leq\,\bar{E}_{oF}\,.
\label{eq:lower bound to fermionic EoF}
\end{align}
For two qubits $E_{oF}=E_{oF}(C)$ is a monotonically increasing function of the concurrence $C$. We propose an analogous functional dependence
of $\bar{E}_{oF}=\bar{E}_{oF}(\bar{C})$ on a parameter $\bar{C}$, that we call ``fermionic concurrence." Evidently, the function $\bar{C}(\varrho)$
is an entanglement monotone that is bounded from below by the usual concurrence $C$. As shown in Ref.~\cite{VerstraeteAudenaertDehaeneDeMoor2001}, the
negativity $\mathcal{N}$ further provides a lower bound to the concurrence, i.e., in our convention of Eq.~(\ref{eq:negativity}), $2\mathcal{N}\leq C$. Consequently, the negativity provides a lower bound to $\bar{C}$, i.e.,
\begin{align}
2\mathcal{N}    &\leq\,C\,\leq\,\bar{C}\,.
\label{eq:negativity and ferm concurrence bounds}
\end{align}
For two modes it is thus at least possible to compute lower bounds to entanglement measures explicitly. It was suggested in Ref.~\cite{WisemanVaccaro2003}
that conventional entanglement measures overestimate the quantum correlations that can physically be extracted from fermionic systems. The operations that can be performed on each single-mode subsystem are limited by (charge) superselection as well. However, we conjecture that the inaccessible entanglement between the fermionic modes can always be swapped to two (uncharged) bosonic modes for which the local bases can be chosen arbitrarily.

\section{Fermionic entanglement beyond two modes}\label{sec:fermionic entanglement beyond 2 modes}

Finally, let us consider the entanglement between more than two fermionic modes. In principle, any measure
of entanglement that is based on entropies of the subsystems is well defined on the fermionic Fock space, as we
have discussed. However, we would like to employ operational measures. Let us therefore start by attempting a
consistent mapping from three fermionic modes to three qubits, in analogy to the two-mode case in
Sec.~\ref{sec:entanglement of fermionic modes}. For simplicity we assume that the modes $\kappa$, $\kappa\pr$, and
$\kappa\prpr$ all have equal charge such that the most general mixed state of these modes can be written as
\begin{align}
    \varrho_{\kappa\kappa\pr\kappa\prpr}   &=\,
    \mu_{1}\,\fket{0}\!\fbra{0}\,+\,
    \mu_{2}\,\fket{\!1_{\kappa\prpr}\!}\!\fbra{\!1_{\kappa\prpr}\!}
    \label{eq:general three-mode fermion state supseselection}\\[1.0mm]
    &\ +\,
    \mu_{3}\,\fket{\!1_{\kappa\pr}\!}\!\fbra{\!1_{\kappa\pr}\!}
    \,+\,
    \mu_{4}\,\fket{\!1_{\kappa\pr}\!}\fket{\!1_{\kappa\prpr}\!}\!\fbra{\!1_{\kappa\prpr}\!}\fbra{\!1_{\kappa\pr}\!}
    \nonumber\\[2.0mm]
    &\ +\,
    \mu_{5}\,\fket{\!1_{\kappa}\!}\!\fbra{\!1_{\kappa}\!}\,+\,
    \mu_{6}\,\fket{\!1_{\kappa}\!}\fket{\!1_{\kappa\prpr}\!}\!\fbra{\!1_{\kappa\prpr}\!}\fbra{\!1_{\kappa}\!}
    \nonumber\\[2.0mm]
    &\ +\,
    \mu_{7}\,\fket{\!1_{\kappa}\!}\fket{\!1_{\kappa\pr}\!}\!\fbra{\!1_{\kappa\pr}\!}\fbra{\!1_{\kappa}\!}
    \nonumber\\[1.5mm]
    &\ +\,
    \mu_{8}\,\fket{\!1_{\kappa}\!}\fket{\!1_{\kappa\pr}\!}\fket{\!1_{\kappa\prpr}\!}\!
        \fbra{\!1_{\kappa\prpr}\!}\fbra{\!1_{\kappa\pr}\!}\fbra{\!1_{\kappa}\!}
    \nonumber\\[1.5mm]
    &\ +\,
    \Bigl(\,
    \nu_{1}\,\fket{\!1_{\kappa\prpr}\!}\!\fbra{\!1_{\kappa\pr}\!}\,+\,
    \nu_{2}\,\fket{\!1_{\kappa\prpr}\!}\!\fbra{\!1_{\kappa}\!}
    \nonumber\\[1.5mm]
    &\ +\,
    \nu_{3}\,\fket{\!1_{\kappa\pr}\!}\!\fbra{\!1_{\kappa}\!}\,+\,
    \nu_{4}\,\fket{\!1_{\kappa\pr}\!}\fket{\!1_{\kappa\prpr}\!}\!\fbra{\!1_{\kappa\prpr}\!}\fbra{\!1_{\kappa}\!}
    \nonumber\\[1.0mm]
    &\ +\,
    \nu_{5}\,\fket{\!1_{\kappa\pr}\!}\fket{\!1_{\kappa\prpr}\!}\!\fbra{\!1_{\kappa\pr}\!}\fbra{\!1_{\kappa}\!}
    \nonumber\\[1.5mm]
    &\ +\,
    \nu_{6}\,\fket{\!1_{\kappa}\!}\fket{\!1_{\kappa\prpr}\!}\!\fbra{\!1_{\kappa\pr}\!}\fbra{\!1_{\kappa}\!}\,+\,\mathrm{H.c.}\,\Bigr)\,.
    \nonumber
\end{align}
The relevant consistency conditions to construct the three different reduced two-mode density matrices
$\varrho_{\kappa\kappa\pr}$, $\varrho_{\kappa\kappa\prpr}$ and $\varrho_{\kappa\pr\kappa\prpr}$ are given by
\begin{subequations}
\label{eq:three mode consistency conditions}
    \begin{align}
        \tr\bigl(\,(b_{\kappa}^{\dagger}b_{\kappa\pr}\,+\,b_{\kappa\pr}^{\dagger}b_{\kappa})\,
        \varrho_{\kappa\kappa\pr\kappa\prpr}\bigr)  &=\,2\,\mathrm{Re}(\nu_{3}\,+\,\nu_{4})\,,
        \label{eq:three mode consistency conditions kappa kappa pr}\\[1.5mm]
        \tr\bigl(\,(b_{\kappa}^{\dagger}b_{\kappa\prpr}\,+\,b_{\kappa\prpr}^{\dagger}b_{\kappa})\,
        \varrho_{\kappa\kappa\pr\kappa\prpr}\bigr)  &=\,2\,\mathrm{Re}(\nu_{2}\,-\,\nu_{5})\,,
        \label{eq:three mode consistency conditions kappa kappa prpr}\\[1.5mm]
        \tr\bigl(\,(b_{\kappa\pr}^{\dagger}b_{\kappa\prpr}\,+\,b_{\kappa\prpr}^{\dagger}b_{\kappa\pr})\,
        \varrho_{\kappa\kappa\pr\kappa\prpr}\bigr)  &=\,2\,\mathrm{Re}(\nu_{1}\,+\,\nu_{6})\,.
        \label{eq:three mode consistency conditions kappa pr kappa prpr}
    \end{align}
\end{subequations}
Again, the correct partial traces are obtained by tracing ``inside out" [see Eq.~(\ref{eq:partial trace offdiagonal ambiguous inside out})].
This is not a coincidence. The prescription for the partial trace to anticommute operators towards the projector of the vacuum state before
eliminating them takes into account the number of anticommutations occurring in computations of the expectation values of
Eq.~(\ref{eq:consistency condition}). A matrix representation of the three-mode state $\varrho_{\kappa\kappa\pr\kappa\prpr}$ is given by
\begin{align}
\varrho_{\kappa\kappa\pr\kappa\prpr}    &=\,
    \begin{pmatrix}
        \mu_{1}     &   0           &   0           &   0           &   0       &   0           &   0       &   0       \\
        0           &   \mu_{2}     &   \nu_{1}     &   0           &   \nu_{2} &   0           &   0       &   0       \\
        0           &   \nu_{1}^{*} &   \mu_{3}     &   0           &   \nu_{3} &   0           &   0       &   0       \\
        0           &   0           &   0           &   \mu_{4}     &   0       &   \nu_{4}     &   \nu_{5} &   0       \\
        0           &   \nu_{2}^{*} &   \nu_{3}^{*} &   0           &   \mu_{5} &   0           &   0       &   0       \\
        0           &   0           &   0           &   \nu_{4}^{*} &   0       &   \mu_{6}     &   \nu_{6} &   0       \\
        0           &   0           &   0           &   \nu_{5}^{*} &   0       &   \nu_{6}^{*} &   \mu_{7} &   0       \\
        0           &   0           &   0           &   0           &   0       &   0           &   0       &   \mu_{8}
    \end{pmatrix}\,.
\label{eq:general three-mode fermion state supseselection matrix rep}
\end{align}
Similar as before, one can try to interpret Eq.~(\ref{eq:general three-mode fermion state supseselection matrix rep}) as a matrix
representation of a three-qubit state and exchange the signs of the basis vectors in the three qubit state such that the consistency
conditions of Eq.~(\ref{eq:three mode consistency conditions}) are met, i.e., opposite signs in front of $\nu_{2}$ and $\nu_{6}$, while the signs
in front of the pairs $\nu_{3},\nu_{4}$ and $\nu_{1},\nu_{6}$ are each the same. This is not possible, even though superselection
rules are respected. This suggests that the superselection rules only coincidentally aid the fermionic qubit mapping for two modes.
They simply force all the problematic coefficients to disappear. However, for more than two modes we find here that a mapping to a
tensor product space cannot be performed consistently in general. Therefore, computing a measure like the negativity to determine the
entanglement between more than two modes appears to be meaningless. Due to the lack of practical alternatives, the minimization over
all states consistent with charge superselection to find $\bar{E}_{oF}$ of Eq.~(\ref{eq:fermionic entanglement of formation}) should
be considered since the restriction of the set of permissable states could make this computation feasible.

\section{conclusion}\label{sec:conclusion}

We have discussed the implementation of fermionic modes as fundamental objects for quantum information tasks. The foundation of this
task is the rigorous construction of the notion of mode subsystems in a fermionic Fock space. We have demonstrated that this can
be achieved despite the absence of a simple tensor product structure. Our simple consistency conditions give a clear picture of
this process, which can be easily executed operationally by performing partial traces ``inside out." Thus we show that fermionic
mode entanglement, quantified by the (fermionic) entanglement of formation or any other function of the eigenvalues of the reduced
states, is indeed a well-defined concept, free of any ambiguities and independent of any superselection rules.

However, problems arise when mappings from the fermionic Fock space to qubit spaces are attempted. We have explicitly demonstrated in
two examples, for two and three modes, that such mappings cannot generally succeed. Only in the limited case where only two modes are
considered and the quantum states obey charge superselection can one meaningfully speak of an equivalence between the two fermionic
modes and two qubits. In this case the application of tools such as the negativity or concurrence is justified. We have argued that
these measures will at least provide a lower bound to genuine measures of fermionic mode entanglement.

Nonetheless, open questions remain. In particular, it is not clear if any operational measures exist for situations beyond two qubits. In
Ref.~\cite{FriisHuberFuentesBruschi2012} witnesses for genuine multipartite entanglement are employed, which are completely compatible with
the framework we have presented here, but these witnesses can only provide lower bounds to entropic entanglement measures.

Finally, we have conjectured that the entanglement in fermionic modes is accessible even in spite of superselection rules that restrict
the possible operations performed on single modes by means of entanglement swapping. The investigation of this question, while beyond the scope
of this article, will certainly be of future interest.


\begin{acknowledgements}
We thank Gerardo Adesso, Samuel~L.~Braunstein, \v{C}aslav Brukner, Fabio Costa, Ivette Fuentes, Marcus Huber, Jorma Louko, Eduardo Mart\'{i}n-Mart\'{i}nez, Miguel Montero, Karoline M{\"u}hlbacher, Carlos Sab\'{i}n, Michael~Skotiniotis, Vlatko Vedral, and Magdalena Zych for useful
discussions and comments. N.~F. acknowledges support from EPSRC (CAF Grant No.~EP/G00496X/2 to I.~F.).
We also want to thank the Perimeter Institute for Theoretical Physics and the organizers of the RQI-N 2012 conference,
where part of this research was conducted.
\end{acknowledgements}

\end{document}